\documentclass[review]{elsarticle}
\usepackage{orcidlink}

\usepackage{hyperref}

\usepackage[color=green!40,textsize=scriptsize]{todonotes}
\usepackage{xargs}                      
\newcommandx{\unsure}[2][1=]{\todo[linecolor=red,backgroundcolor=red!25,bordercolor=red,#1]{#2}}
\newcommandx{\toself}[2][1=]{\todo[linecolor=blue,backgroundcolor=blue!25,bordercolor=blue,#1]{#2}}

\newcounter{cases}
\newcounter{subcases}
\newenvironment{mycases}
  {%
    \setcounter{cases}{0}%
    \def\case
      {%
        \par\noindent
        \refstepcounter{cases}%
        \textbf{Case \thecases.}
      }%
  }
  {%
    \par
  }
\newenvironment{subcases}
  {%
    \setcounter{subcases}{0}%
    \def\subcase
      {%
        \par\noindent
        \refstepcounter{subcases}%
        \textit{Subcase (\thesubcases):}
      }%
  }
  {%
  }
\renewcommand*\thecases{\arabic{cases}}
\renewcommand*\thesubcases{\roman{subcases}}
\colorlet{MyColorOne}{green!50}
\colorlet{MyColorOneDark}{green!50!black}
\colorlet{MyColorTwo}{blue!60}
\usepackage{amssymb,amsthm,enumitem}
\newtheorem{theorem}{Theorem}  
\newtheorem{lemma}{Lemma}

\newtheorem{corollary}{Corollary}
\newtheorem{definition}{Definition}
\newtheorem{observation}{Observation}
\usepackage{amsmath}
\usepackage[noabbrev,capitalise,nameinlink]{cleveref}
\Crefname{conjecture}{Conjecture}{Conjectures}
\Crefname{algorithm}{Algorithm}{Algorithm}
\usepackage{mathtools}

\usepackage{tabularx,lipsum,environ,mathrsfs,caption, tikz}
\usetikzlibrary{decorations.shapes,backgrounds}
\usepackage{comment}
\usepackage[ruled,vlined]{algorithm2e} 
\usepackage{booktabs,pdfpages,array,etoolbox}
\preto\tabular{\setcounter{magicrownumbers}{0}}
\newcounter{magicrownumbers}

\def\tw{\text{tw}}
\definecolor{orcidlogocol}{HTML}{A6CE39}

\usepackage{subcaption}
\usepackage{graphicx}

\journal{Discrete Mathematics}

\begin{document}

\begin{frontmatter}

\title{
Safe Edges: A Study of Triangulation in Fill-in and Tree-Width Problems
}


\author[UoP]{Janka Chleb\'\i{}kov\'a}
\ead{janka.chlebikova@port.ac.uk}

\author[UoP]{Mani Ghahremani\corref{mycorrespondingauthor}
\orcidlink{0000-0001-7693-2048}
}
\cortext[mycorrespondingauthor]{Corresponding author}
\ead{mani.ghahremani@port.ac.uk}

\address[UoP]{School of Computing, University of Portsmouth, United Kingdom}

\begin{abstract}
This paper considers two well-studied problems \textsc{Minimum Fill-In} (\textsc{Min Fill-In}) and \textsc{Treewidth}.
Since both problems are \textsf{NP}-hard, various reduction rules simplifying an input graph have been intensively studied to better understand the structural properties relevant to these problems.
Bodlaender at el.\ \cite{minDegree} introduced the concept of a safe edge that is included in a solution of the \textsc{Minimum Fill-In} problem and showed some initial results.
In this paper, we extend their result and prove a new condition for an edge set to be safe.
This in turn helps us to construct a novel reduction tool for \textsc{Min Fill-In} that we use to answer other questions related to the problem.

In this paper, we also study another interesting research question: Whether there exists a triangulation that answers both problems \textsc{Min Fill-In} and \textsc{Treewidth}.
To formalise our study, we introduce a new parameter reflecting a distance of triangulations optimising both problems.
We present some initial results regarding this parameter and study graph classes where both problems can be solved with one triangulation.
\end{abstract}


\end{frontmatter}


\section{Introduction}
{
The minimum fill-in and treewidth of graphs are well studied graph parameters with many practical applications \cite{minimalTriangulationSurvey}.
The minimum fill-in of a graph $G$, $\textit{mfi}(G)$, is the minimum number of edges that triangulate it.
The treewidth parameter, $\textit{tw}(G)$, is equal to the minimum clique size over all triangulations of the input graph, subtracted by one.
This naturally leads to the interesting  question of whether there exists a triangulation that minimizes both parameters at the same time. 
Clearly, one such triangulation answers both NP-hard problems \textsc{Minimum Fill-In}  and \textsc{Treewidth} simultaneously.

To formalise our study, we introduce a new parameter $\tau$.

\begin{definition}\label{def:tau}
Let $G$ be a graph and $H$ be a minimum triangulation of $G$ chosen so that its clique size is minimized.
Define $\tau(G)=\tw(H)-\tw(G)$.
\end{definition}

Knowing that $\tau(G)=0$ for a graph $G$, implies that there exists a minimum triangulation of $G$ among those determining the treewidth and vice versa.
Conversely, distinct graph classes have been constructed where one triangulation cannot solve both problems \cite{kloks1994treewidth,manciniPhd,Dar2019}.
In such graphs the $\tau$ parameter can be arbitrarily large; see, for example, the construction from \cite{myPhd} that explicitly uses the parameter $\tau$.
To study this parameter, we design methods that can be useful in determining the value of $\tau$ as well as the minimum fill-in of graphs.

In \cref{sec:defs} we provide the necessary definitions and some properties of elimination orderings, as well as a reduction rule that can be useful for the \textsc{Minimum Fill-In} problem.
Next in \cref{sec:safe}, we extend the notion of safe edges while constructing a minimum triangulation of a graph, a concept that was introduced in \cite{minDegree}.
In \cref{sec:tfm:Ktrees}, we present initial study of the $\tau$ parameter on graphs with low treewidth and graphs with the same value of treewidth and vertex connectivity.
}

\section{Preliminaries}\label{sec:defs}

{
The used terminology  is consistent with textbooks such as \cite{Introduction2,Introduction1}.
We assume all graphs $G=(V,E)$ to be finite, simple, connected, and undirected.
For a vertex $v\in V$, the neighbourhood of $v$, $N(v)$, refers to the set of vertices adjacent to $v$.
The closed neighbourhood of $v$ is $N[v]=N(v)\cup \{v\}$
and the degree of $v$ is $\textit{deg}(v)=|N(v)|$.
For a vertex set $X\subseteq V$, we use $G[X]$ to denote the graph induced on $X$.
A vertex $v\in V$ is \textit{universal} in $G$ if $N(v)=V\setminus\{v\}$.
Given a graph $G$ and a vertex set $X\subseteq V$, $G-X$ refers to the graph obtained from $G$ after the \textit{removal} of $X$ and we use $G-v$, as a shorthand for $G-\{v\}$, for a vertex $v\in V$.
The set of \textit{missing edges}, or \textit{fill edges}, of a vertex set $X\subseteq V$ is defined as $\textit{fill}(X)=\{uv|u,v\in X, uv\notin E\}$.
For a set of missing edges $F$ we refer to the graph obtained by \textit{addition} of $F$ to $G$ as $G\oplus F$, which is a shorthand for the graph $(V, E\cup F)$. 
Given a vertex $v\in V$, \textit{elimination} of $v$ equates to the addition of the missing edges in neighbourhood of $v$,  $\textit{fill}(N(v))$, to $G$ and then removal of $v$ in the resulting graph, formally $(G - v) \oplus \textit{fill}(N(v))$.
A vertex set $Y \subseteq V$ is \textit{almost clique} in $G$ if there exists a vertex $x\in Y$ such that $Y\setminus \{x\}$ is a clique.
A vertex is called \textit{(almost) simplicial} if its neighbourhood is a (almost) clique.

A \textit{chord} is an edge connecting two non-consecutive vertices of a cycle, and a graph is \textit{chordal} iff every cycle with at least four vertices has a chord.
Given a graph $G$, a chordal supergraph $H=(V, E\cup F)$ is called a \textit{triangulation} of $G$, where $F$ is a set of added chords.
A triangulation $H'=(V,E\cup F')$ is a \textit{minimal triangulation} of $G$ if for every proper subset of edges $F''\subset F'$, the graph $H''=(V,E\cup F'')$ is not chordal.
Similarly, a triangulation $H^*=(V,E\cup F^*)$ is a \textit{minimum triangulation} if $|F^*|$ is the minimum number of fill edges that triangulate $G$.
The \textit{minimum fill-in} parameter of $G$ is defined as $\textit{mfi}(G)=|F^*|$.

Given the graph $G$, the \textsc{Minimum Fill-In} problem is answered by determining the value of $\textit{mfi}(G)$.
The \textsc{Treewidth} problem can be answered by finding a triangulation $H$ of the input graph $G$ where its clique size, $\omega(H)$, is minimized.
Then the treewidth of the graph, $\textit{tw}(G)$, is defined as $\textit{tw}(G)=\omega(H)-1$.
However, the treewidth parameter can be defined in many ways; we recommend \cite{bodLaenderTW} for a complete introduction.
}

{
The \textit{vertex-connectivity} (or simply connectivity) of $G$, $\kappa(G)$, is the smallest number of vertices whose removal disconnects $G$.
As an exception, the connectivity of a complete graph on $k$ vertices is defined as $\kappa(K_k)=k-1$, for $k\geq 1$.
Given two non-adjacent vertices $a,b\in V$, a vertex set $S\subseteq V\setminus\{a,b\}$ is an $a,b$-separator in $G$, or a separator for short, if $a$ and $b$ are contained in different components of $G-S$.
For a separator $S\subset V$, we use $\mathcal{C}(S)$ to refer to the set of connected components of $G-S$.
An $a,b$-separator $S$ is a \textit{minimal} separator in $G$, if no proper subset $S'\subset S$ is also an $a,b$-separator.
A (minimal) separator $S$ in $G$ is a clique (minimal) separator if $S$ is a clique in $G$.

The following is a well-known property of clique minimal separators (discussed in \cite{cliqueminimalseparatordecomposition}) that is essential to prove some of our later findings in this paper.

\begin{theorem}
\label{lem:cliqueMinimalSeparator}
Let $G=(V,E)$ be a graph and $S$ a clique minimal separator in $G$. 
Then:
$$\textit{mfi}(G)=\sum\nolimits_{C\in \mathcal{C}(S)}\textit{mfi}(G[S\cup V_C]) $$
\end{theorem}
}

{
An \textit{elimination ordering} $\alpha$ over a graph $G$ is a bijection $\alpha:\{1,\dots,|V|\} \rightarrow V$.
It is shown in \cite{fulkerson1965} that any given elimination ordering $\alpha$ over a graph $G$ can be used to define a triangulation $G^+_\alpha$ as outlined below:

\begin{definition}\label{def:elimGalpha}
Given an elimination ordering $\alpha$ over the graph $G=(V,E)$ we define the supergraph $G^+_\alpha=(V,E\cup F)$ using the sequence of graphs $G_0,\dots,G_{|V|}$: 
\begin{itemize}
    \item[$\bullet$] Let $G_0=G$.
    \item[$\bullet$] For each step $i\in\{1,\dots,|V|\}$: 
    let $\alpha(i)$ be a vertex of $G_{i-1}$ and denote
    $F_i=\textit{fill}_{G_{i-1}}\,(\,N_{G_{i-1}}\,(\alpha(i))\,)$. 
    Then define the graph obtained from $G_{i-1}$ after eliminating $\alpha(i)$ as $G_{i}=(G_{i-1}\oplus F_i)-\alpha(i)$.
\end{itemize}  

\noindent Finally, define $F=\bigcup_{i\in\{1,\dots,|V|\}}F_i$.
\end{definition}

Given a graph $G=(V,E)$ and an elimination ordering $\alpha$ over $G$,
we write $G^\alpha_i$ to refer to the graph obtained at the step $i\in \{1,\dots,|V|\}$ of $\alpha$.
For a vertex $v\in V$, we define $\textit{madj}^+_\alpha(v)=\{u|u\in N_{G^+_\alpha}(v),\alpha^{-1}(u)>\alpha^{-1}(v)\}$.
Following \cref{def:elimGalpha}, we note that $\textit{madj}^+_\alpha\,(\alpha(i))=N_{G_{i-1}}\,(\alpha(i))$ for any step $i\in\{1,\dots,|V|\}$.

An elimination ordering $\alpha$ over the graph $G$ is a \textit{perfect elimination ordering}, peo for short, if for every step $i\in\{1,\dots,|V|\}$, $\alpha(i)$ is a simplicial vertex in $G_{i-1}$, therefore $G^+_{\alpha}=G$.
An elimination ordering $\alpha$ over $G$ is \textit{minimal} if $G^+_{\alpha}$ is a minimal triangulation of $G$.
In \cite{eliminationorderingMinimaltraingulation} it has been shown that for every minimal triangulation $H$ of the graph $G$, there exists a minimal elimination ordering $\alpha$ over $G$ where $G^+_\alpha=H$.

Since every minimum triangulation is also a minimal triangulation, the \textsc{Minimum Fill-In} problem over the graph $G$ can be answered using an elimination ordering $\alpha$ which constructs a minimum triangulation $G^+_\alpha$.
Therefore, if $\alpha$ is a \textit{minimum elimination ordering}, then $\sum_{v\in V}|\textit{madj}^+_\alpha(v)|=\textit{mfi}(G)+|E|$.
Similarly, there must exist an elimination ordering $\beta$ such that $\max_{v\in V}|\textit{madj}^+_\beta(v)|=
\textit{tw}(G^+_\beta)=\textit{tw}(G)$ solving the \textsc{Treewidth} problem.
Furthermore, if there exists a minimum elimination ordering $\gamma$ over the graph $G$ such that $\max_{v\in V}|\textit{madj}^+_\gamma(v)|=\textit{tw}(G)$ we have that $\tau(G)=0$.

{
{
The following lemma demonstrates an easy-to-prove property of elimination ordering:

\begin{lemma}\label{cor:madjNonAdjacent}
Let $G=(V,E)$ be a graph, $\alpha$ an elimination ordering and 
let $v\in V$ be a vertex eliminated at step $i\in\{1,\dots,|V|\}$. Then  $N_{G_{i-1}}(v)=N_G(v)$ if $\alpha(j)\notin N_G(v)$ for every  step $j$, $j<i$.
\end{lemma}
} 

{
Next, we present a lemma discussing the vertex connectivity of graphs in regard to the operation of vertex elimination:

\begin{lemma}\label{cor:tw=kConnected}
Let $G=(V,E)$ be a connected graph with the vertex connectivity of $k$ and $\alpha$ any elimination ordering over $G$. 
Then the following statements are true (where $G_0=G$ and $G_i$ is the graph obtained at the step $i$ of elimination ordering $\alpha$ as defined in \cref{def:elimGalpha}):
\begin{itemize}
    \item[(i)] For every step $i\in\{1,\dots, |V|-k-1\}$ of the elimination ordering $\alpha$, $|\textit{madj}^+_\alpha(\alpha(i))|\geq \kappa(G_i)\geq k$.
    
    \item[(ii)] $G_{|V|-k-1}$ is a complete graph, hence $\sum_{i\in\{|V|-k,\dots, |V|\}}
    |\textit{madj}^+_\alpha(\alpha(i))|={\frac{k(k+1)}{2}}$.
\end{itemize}

Consequently $|E|+\textit{mfi}(G)\geq k(|V|-k)+\frac{k(k-1)}{2}$.

\end{lemma}

\begin{proof}
In order to prove the statement (i), it suffices to show that $\kappa(G_i)\geq \kappa(G_{i-1})$ for any step $i\in\{1,\dots, |V|-k-1\}$.
Since $G$ is connected, we have $k=\kappa(G_0)\geq 1$.

Suppose for a contradiction that there exists a step $i\in\{1,\dots, |V|-k-1\}$ in which $\kappa(G_i)<\kappa(G_{i-1})$.
Then we can fix a subset $S\subset V_{G_i}$ where $\kappa(G_i)\leq |S|< \kappa(G_{i-1})$ and $G_i-S$ is not connected.
Consequently, we can fix two vertices $u_1,u_2\in V_{G_i}$ such that there are no $u_1,u_2$-paths in $G_i-S$.
Since $|S|< \kappa(G_{i-1})$, we can let $P$ be the shortest $u_1,u_2$-path in $G_{i-1}-S$.
Clearly, $P$ must contain the vertex $v$ as an internal vertex, otherwise $P$ is a $u_1,u_2$-path in $G_i-S$.
Then $(y, v, w)$ is a subpath of $P$ for two vertices $y,w\in N_{G_{i-1}}(v)$ (where $u_1, u_2$ are not necessarily different from $y, w$).

Since $N_{G_{i-1}}(v)$ becomes a clique in $G_i$ after elimination of vertex $\alpha(i)$, we can define $P'$ as the path obtained from $P$ by replacing the subpath $(y, v, w)$ with the edge $yw$. 
Then $P'$ is by definition a $u_1,u_2$-path in $G_i-S$ which results in a contradiction proving statement (i).
Combining this and the fact that $G_{|V|-k-1}$ is a graph with $k+1$ vertices and connectivity of $k$, proves statement (ii).
\end{proof}
} 

} 
}
\section{Edges That Can Be Safely Added}\label{sec:safe}

{
When constructing a minimum triangulation of a graph in \cite{minDegree}, Bodlaender et al.\, provided a criterion for a set of missing edges to be safely added.
In this section, we extend their results by providing another characteristic of such edges.
Our result will be used in \cref{lem:safeVertexRepeatedly} to define a reduction rule for the \textsc{Minimum Fill-In} problem.
This theorem outlines a set of vertices of the input graph that can be safely eliminated at the start.
}

{
\begin{definition}\label{def:safeEdges}
Given a graph $G=(V,E)$, a set of missing edges $F'\subseteq \textit{fill}(V)$ is safe to add in $G$ if there exists a minimum triangulation $H=(V,E\cup F)$ of $G$ where $F'\subseteq F$.
\end{definition}

The following lemmas are easy to prove properties of edges that are safe to add (formal proofs can be found in \cite{myPhd}):

\begin{lemma}\label{lem:safeEdgesFirst}
Let $G$ be a graph, and $F'$ a set of safe edges to add in $G$.
Then every minimum triangulation $H'$ of $G\oplus F'$ is also a minimum triangulation of~$G$.
\end{lemma}

\begin{lemma}\label{lem:safeRepeatedly}
Let $F'$ be a set of safe edges to add in $G$ and $F''$ be a set of safe edges to add in $G\oplus F'$. 
Then the edge set $F'\cup F''$ is safe to add in $G$.
\end{lemma}

Combining \cref{lem:safeEdgesFirst} and \cref{lem:safeRepeatedly}, we can conclude that a minimum triangulation of a graph can be constructed by a repeated addition of a set of safe edges and subsequently constructing a minimum triangulation of the resulting graph.
}

{
The following theorem by Bodlaender et al. (\cite{minDegree}) 
presents a condition for an edge set to be safe to add.
Notice that the sets of edges it covers must have only one missing edge.

\begin{theorem}[\cite{minDegree}]\label{lem:bodlaenderSafeEdges}
Let $S$ be a minimal separator of the graph $G=(V,E)$ where $|\textit{fill}(S)|=1$ and $S\subseteq N(v)$ for a vertex $v\in V\setminus S$.
Then $\textit{fill}(S)$ is safe to add in $G$.
\end{theorem}
}

{
To the best of our knowledge, this remains the only known criterion for edges that are safe to add. 
In the following theorem, we prove a new condition for a set of edges to be safe to add.
We note that, an almost simplicial vertex has a degree of at least two.
Therefore, the following theorem is only applicable to graphs with connectivity of at least two.

\begin{theorem}\label{lem:degkAlmostSimplicial}
Let $G=(V,E)$ be a graph, and $v\in V$ be an almost simplicial vertex where $\textit{deg}(v)=\kappa(G)$.
Then $\textit{fill}(N_G(v))$ is safe to add.
\end{theorem}

\begin{proof}
Let $k=\textit{deg}(v)$ and $N_G(v)=\{w_1,\dots,w_k\}$.
As noted above, the fact that $v$ is almost simplicial implies that $k\geq 2$ and we can assume that $N_G(v)\setminus\{w_1\}$ is a clique in $G$.

To prove the statement of the theorem, 
suppose for a contradiction that in every minimum triangulation $H=(V,E\cup F)$ of $G$, $N_G(v)$ is not a clique.
Fix one such minimum triangulation $H$ and observe that because $H$ is a supergraph of $G$ (obtained by only adding edges), we have $\kappa(H)\geq k=\kappa(G)$ and $N_G(v)\subseteq N_H(v)$.
Define the set of edges $I_v\subseteq F$ as  
$I_v=\{vu\; |\; u\in N_H(v)\setminus N_G(v)\}$.
In the first stage of the proof we establish that for every missing edge $w_1w_i\in \textit{fill}_H\,(N_G(v))$, for $i\in\{2,\dots,k\}$, there exists a unique vertex $u_i$ such that $vu_i\in I_v$ and $w_1u_i\in E\cup F$.

We prove this by induction on the size of
$\textit{fill}_H\,(N_G(v))$ starting with the base case: $|\textit{fill}_H\,(N_G(v))|=1$.
Assume that $\textit{fill}_H\,(N_G(v))=\{w_1w_2\}$.
As discussed before, $\kappa(H)\geq k$ and $|N_G[v]|=k+1$ hence the graph $H-(N_G[v]\setminus\{w_1,w_2\})$ remains connected.
Let $P_{w_1,w_2}=(w_1,u_2^1,\dots,u_2^l,w_2)$ for $l\geq 1$ be the shortest $w_1,w_2$-path in the graph $H-(N_G[v]\setminus\{w_1,w_2\})$, where obviously $u_2^j\notin N_G[v]$, for all $j\in \{1,\dots, l\}$.
Since $P_{w_1,w_2}$ is the shortest path between $w_1,w_2$ in $H-(N_G[v]\setminus\{w_1,w_2\})$ and $w_1w_2\notin E\cup F$, the cycle $C_{w_1,w_2}=(v,w_1,u_2^1,\dots,u_2^l,w_2,v)$ in $H$ does not have any chords with both vertices of the vertex set $\{w_1,u_2^1,\dots,u_2^l,w_2\}$.
Then we necessarily have that $vu_2^j\in F$ for every $j\in \{1,\dots, l\}$ otherwise $C_{w_1,w_2}$ is a chordless cycle in $H$ contradicting the assumption that $H$ is a triangulation of $G$.
Then in this case, select $u_2 = u_2^1\in N_H(v)\setminus N_G(v)$ be the unique vertex corresponding to the missing edge $w_1w_2\in \textit{fill}_H\,(N_G(v))$.
By selection, we have $vu_2\in I_v$ and $w_1u_2\in E\cup F$ as required.

Next, suppose that $|\textit{fill}_H\,(N_G(v))|\geq 2$.
Since $|\textit{fill}_G\,(N_G(v))|\geq
|\textit{fill}_H\,(N_G(v))|$ and $|\textit{fill}_G\,(N_G(v))|\leq k-1$, we have $|\textit{fill}_H\,(N_G(v))|\leq k-1$.
Assume that our claim holds for any $c=|\textit{fill}_H\,(N_G(v))|$, where $1\leq c\leq k-2$, and now we extend it to the case where $|\textit{fill}_H\,(N_G(v))|=c+1$.
Let $\textit{fill}_H\,(N_G(v))=\{w_1w_2,\dots,w_1w_{c+2}\}$.
Therefore, we are supposing that for every missing edge in $\{w_1w_2,\dots,w_1w_{c+1}\}\subset \textit{fill}_H\,(N_G(v))$ there exists a unique vertex in $\{u_2,\dots,u_{c+1}\}$ such that $\{vu_2,\dots,vu_{c+1}\}\subset I_v$ and $\{w_1u_2,\dots,w_1u_{c+1}\}\subset E\cup F$.

Now we need to prove that there exists a unique vertex $u_{c+2}\in N_H(v)\setminus N_G(v)$ associated with the missing edge $w_1w_{c+2}\in \textit{fill}_H\,(N_G(v))$ where $w_1u_{c+2}\in E\cup F$.
Since $\kappa(G)\geq k$, the graph $H-\{u_2,\dots,u_{c+1},v,w_{c+3},\dots,w_{k}\}$ remains connected.
Note that in the case where $c=k-2$, we refer to the graph $H-\{u_2,\dots,u_{c+1},v\}$.
Let $P_{w_1,w_{c+2}}$ be the shortest $w_1,w_{c+2}$-path in this graph.
By selection and the fact that $\{w_2,\dots,w_{c+2}\}$ is a clique in $H$, $P_{w_1,w_{c+2}}$ can contain at most one internal vertex from $\{w_2,\dots,w_{c+1}\}$.
Thus, $P_{w_1,w_{c+2}}$ is either $(w_1,u_{c+2}^1,\dots,u_{c+2}^l,w_{c+2})$ or $(w_1,u_{c+2}^1,\dots,u_{c+2}^l,w_i,w_{c+2})$ where $w_i\in\{w_2,\dots,w_{c+1}\}$ and $l\geq 1$.
Also following the selection of $P_{w_1,w_{c+2}}$, for every $j\in\{1,\dots,l\}$ we have $u_{c+2}^j\notin N_G[v]$.
Define $C_{w_1,w_{c+2}}$ to be the cycle in $H$ composed of the paths $P_{w_1,w_{c+2}}$ and $(w_1,v,w_{c+2})$.
Since $P_{w_1,w_{c+2}}$ is the shortest path between $w_1,w_{c+2}$ in 
$H-\{u_2,\dots,u_{c+2},v,w_{c+3},\dots,w_k\}$ and $w_1w_{c+2}\notin E\cup F$, the cycle $C_{w_1,w_{c+2}}$ does not have chords with both endpoints in the path $P_{w_1,w_{c+2}}$.
Therefore, we must have $vu_{c+2}^j\in I_v$ for every $j\in \{1,\dots,l\}$ otherwise $C_{w_1,w_{c+2}}$ is a chordless cycle in $H$, a contradiction.
By selection of the path $P_{w_1,w_{c+2}}$ we have $u_{c+2}^1\notin \{u_2,\dots,u_{c+1}\}$, and $w_1u_{c+2}^1\in E\cup F$ therefore, we let $u_{c+2}=u_{c+2}^1$ be the unique vertex in $N_H(v)\setminus N_G(v)$ associated with the missing edge $w_1w_{c+2}\in\textit{fill}_{H}\,(N_G(v))$.
This proves our claim that for every missing edge $w_1w_i\in \textit{fill}_H\,(N_G(v))$, $i\in\{2,\dots,k\}$, there exists a unique vertex $u_i$ such that $vu_i\in I_v$ and $w_1u_i\in E\cup F$.
This will be needed shortly.

For the next stage of the proof, we define the set of chords
$F'=(F\setminus I_v)\cup \{w_1u|vu\in I_v\}\cup \textit{fill}_H\,(N_G(v))$.
We know from the first part of the proof that $|I_v|\geq |\textit{fill}_H\,(N_G(v))|$.
By definition, every missing edge $w_1w_i\in\textit{fill}_H\,(N_G(v))$ added in $F'$ replaces an edge $vu\in I_v$ that is removed from $F$ and the corresponding edge $w_1u$ already exists in $E\cup F$ (hence $w_1u\notin (F'\setminus F)$).
We then replace all the remaining edges $vu'\in I_v$ with an edge $w_1u'$ in $F'$ (unless $w_1u'$ already exists).
As a result, we have that $|F'|\leq |F|$.

Next, define the supergraph $H'=(V,E\cup F')$.
Note that by the construction of $F'$, $N_{H'}(v)=N_G(v)$ and $N_{H'}(v)$ is a clique in $H'$.
Now, let us prove that $H'$ is chordal.
Suppose for a contradiction that there exists a chordless cycle in $H'$.
Let $C'$ be the shortest chordless cycle in $H'$.
Since $N_G(v)=\{w_1,\dots,w_k\}$ is a clique minimal separator in $H'$, separating $v$ from the rest of the graph, we know that $C'$ must be entirely contained in the graph $H'-v$ and can contain at most two vertices from $\{w_1,\dots,w_k\}$.
In the following, we will consider all possibilities and show that every case proves the existence of a chordless cycle in the triangulation $H$, a contradiction.

\renewcommand*\thecases{\Alph{cases}}
\begin{mycases}
    \case 
    {
    Suppose that $w_1$ is not a vertex of the cycle $C'$.
    Hence $C'$ cannot contain any edges from 
    $\textit{fill}_H\,(N_G(v))$ that are added in $H'$.
    This implies that $C'$ is a cycle in $H$ and since it does not contain the vertex $v$, none of the removed edges $I_v$ are chords in $C'$.
    This concludes that $C'$ is a chordless cycle in $H$.
    }
    
    \case
    {
    Suppose that $C'$ only contains the vertex $w_1$ from $\{w_1,\dots,w_k\}$. 
    Let $C'=(w_1,y_1,\dots,y_l,w_1)$ where $l\geq 3$.
    As $C'$ is a chordless cycle, it can contain at most two edges from the set of edges $\{w_1u|vu\in I_v\}$ that are added in $H'$.
    Clearly, these edges may only be $w_1y_1$ or $w_1y_l$.
    Thus, we need to consider the following complementary subcases and prove that each leads to the existence of a chordless cycle in $H$:
    }
    
    \begin{subcases}
        \subcase 
        {Suppose that none of the edges $w_1y_1,w_1y_l$ are added in $H'$.
        Consequently, $C'$ is a cycle in $H$ and because it does not contain the vertex $v$, the removeded $I_v$ are not chords in $C'$.
        As a result, $C'$ is a chordless cycle in $H$.
        }
        
        \subcase 
        {Suppose that $w_1y_1$ is the only edge between $w_1y_1,w_1y_l$ that is added in $H'$ (the other case is symmetrical).
        By construction, $w_1y_1$ in $H'$ replaces $vy_1$ in $H$.
        Consider the cycle $C=(v,y_1,\dots,y_l,w_1,v)$ in $H$.
        Since $C'$ is chordless in $H'$ and only the edges from $I_v$ are removed from $H$, the cycle $C$ can not have any chords between the vertices $y_1,\dots,y_l,w_1$ in $H$.
        Furthermore, for every $i\in \{2,\dots,l\}$, $vy_i\notin I_v$ otherwise, by our construction $C'$ would have the chord $w_1y_i$ in $H'$, contradicting the assumption that $C'$ is chordless.
        This concludes that $C$ is a chordless cycle in $H$, as desired.}
        
        \subcase 
        {Suppose that both edges $w_1y_1,w_1y_l$ are added in $H'$.
        By construction, $w_1y_1,w_1y_l$ replace the edges $vy_1,vy_l\in I_v$ in $H$.
        Then define the cycle $C=(v,y_1,\dots,y_l,v)$ in $H$.
        Similar to the previous subcase, the cycle $C$ cannot have chords between the vertices $y_1,\dots,y_l$ in $H$.
        Additionally, for every $i\in \{2,\dots,l-1\}$ we have $vy_i\notin I_v$ otherwise $w_1y_i$ is a chord in $C'$.
        This proves that $C$ is a chordless cycle in $H$.}
    \end{subcases}
    
    \case
    {
    Suppose that $C'$ contains the vertex $w_1$ and a vertex $w_i\in \{w_2,\dots,w_k\}$.
    Since $\{w_1,\dots,w_k\}$ is a clique in $H'$, $w_1w_i$ must be an edge in the cycle $C'$.
    Let $C'=(w_1,w_i,y_1,\dots,y_l,w_1)$ where $l\geq 2$.
    
    If $w_1w_i$ was already an edge in $H$, then $C'$ is a cycle in $H$.
    Additionally, $C'$ would be a chordless cycle in $H$ as none of the removed edges $I_v$ would belong to $C'$.
    Therefore we can further assume that $w_1w_i\in \textit{fill}_H\,(N_G(v))$.
    
    $w_1y_l$ can potentially be one of the edges that are added in $H'$ so let us consider two complementary cases and define the cycle $C$ in $H$ accordingly: If $vy_l\in I_v$ we define $C=(v,w_i,y_1,\dots,y_l,v)$ otherwise if $vy_l\notin I_v$ let $C=(w_1,v,w_i,y_1,\dots,y_l,w_1)$.
    In either case, $C$ cannot have $vy_j$ as a chord in $H$ for any $j\in\{1,\dots,l-1\}$ since this would imply that the cycle $C'$ has the chord $w_1y_j$ in $H'$.
    Furthermore, going by the assumption that $C'$ is a chordless cycle in $H'$,  $C$ cannot have any chords between vertices in $\{w_1,w_i,y_1,\dots,y_l\}$.
    This concludes that $C$ is a chordless cycle in $H$.
    }
\end{mycases}

This proves that $H'$ is a chordal supergraph of $G$.
Combined with the fact that $|F'|\leq |F|$ we find that $H'$ is also a minimum triangulation of $G$ and by construction $\textit{fill}_G(N_G(v))\subseteq F'$.
Consequently, by \cref{def:safeEdges} we can state that $\textit{fill}_G(N_G(v))$ is safe to add in $G$.
\end{proof}

By construction of the minimum triangulation $H'=(V,E\cup F')$ (in the proof of \cref{lem:degkAlmostSimplicial}), $N_{H'}(v)=N_G(v)$ in addition to $\textit{fill}_G(N_G(v))\subseteq F'$.
In other words, no edges are added in $H'$ with an endpoint in~$v$.
}

{
In what follows, given a graph $G=(V,E)$ and a vertex $v\in V$, we use 
Algorithm~\ref{alg:Fv} to
define a set of  edges $F^v$, $F^v\subseteq \textit{fill}_G\,(N_G(v))$, used later  in~\cref{lem:safeVertexRepeatedly} in a reduction rule for the \textsc{Minimum Fill-In} problem.

\begin{algorithm}[H]
\SetAlgoLined
\KwIn{A graph $G=(V,E)$ and a vertex $v\in V$}
\KwOut{The set of edges $F^v\subseteq \textit{fill}\,(N(v))$}
    Initially let $F^v=\emptyset$\;
    \While{there exists a minimal separator $S$ in the graph $G\oplus F^v$ s.t.\ $S\subseteq N_G(v)$ and $|\textit{fill}_{G\oplus F^v}(S)|=1$}
    {
    Fix one such minimal separator $S\in \mathcal{S}_{G\oplus F^v}$\;
    Add $\textit{fill}_{G\oplus F^v}\,(S)$ to $F^v$\; 
    }
    \caption{Definition of the set $F^v$}
    \label{alg:Fv}
\end{algorithm}
\vspace{1em}

Notice that $F^v$ can be possibly empty and that the order in which the minimal separators are selected in Algorithm \ref{alg:Fv} is not important.
\hyperref[fig:Fv]{Figure 1} demonstrates the steps of Algorithm \ref{alg:Fv} with an example.

\begin{observation}\label{ob:1}
Combining \cref{lem:safeRepeatedly} and a repeated application of \cref{lem:bodlaenderSafeEdges}, it is easy to see that for any vertex $v$ the set $F^v$ is safe to add to $G$.
\end{observation}

\usetikzlibrary{shapes, backgrounds}
\tikzset{
comp/.style={draw=black,ellipse,minimum width=10pt, inner sep=3pt, text width=4mm, text height=2 mm,  align=center, fill=white},
invisible/.style={draw,circle,fill=white,minimum size=0pt,inner sep=0pt},
norm/.style={draw,circle,fill=black,inner sep=0pt}
}
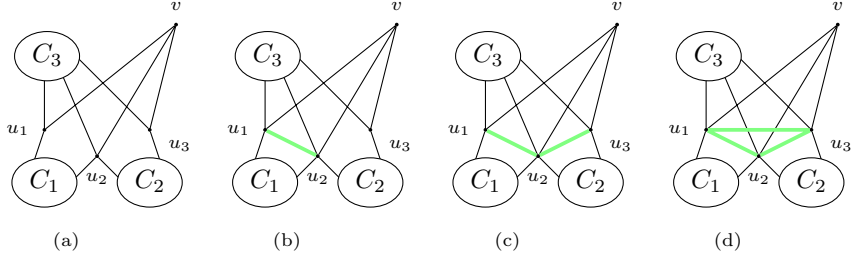
\begin{figure}
    \begin{subfigure}[b]{0.15\textwidth}
        \begin{tikzpicture}[scale=0.35, auto,swap]
            {
            \node [comp] (c1) at (-2,-4) {$C_1$};
            \node [comp] (c2) at (2,-4) {$C_2$};
            \node [comp] (c3) at (-1.90,0.80) {$C_3$};
            }; 
            \begin{scope}[on background layer]
                \draw 
                (3,2) node (v) [norm, label={[label distance=0.05cm]90:\scriptsize $v$}] {}
                (-2,-2) node (u1) [norm, label={[label distance=0.05cm]180:\scriptsize $u_1$}] {}
                (0,-3) node (u2) [norm, label={[label distance=0.08cm]270:\scriptsize $u_2$}] {}
                (2,-2) node (u3) [norm, label={[label distance=0.1cm]-5:\scriptsize $u_3$}] {}
                
                (-2.5,-3.5) node (u12v1) [invisible] {}
                (-1,-4) node (u12v2) [invisible] {}
                
                (2.5,-3.5) node (u32v1) [invisible] {}
                (1,-4) node (u32v2) [invisible] {}
                
                (-1,1) node (u22v1) [invisible] {}
                (-2,0) node (u22v2) [invisible] {}
                (-1.5,0.5) node (u22v3) [invisible] {}
                
                (u1)--(u12v1)
                (u2)--(u12v2)
                
                (u2)--(u32v2)
                (u3)--(u32v1)
                
                (u1)--(u22v2)
                (u3)--(u22v1)
                (u2)--(u22v3)
                
                (v)--(u1)
                (v)--(u2)
                (v)--(u3)
                ;   
                \end{scope}
        \end{tikzpicture}
        \caption{}
        \label{fig:Fv:G}
    \end{subfigure}
    \hspace{2.5em}
    \begin{subfigure}[b]{0.15\textwidth}
        \begin{tikzpicture}[scale=0.35, auto,swap]
            {
            \node [comp] (c1) at (-2,-4) {$C_1$};
            \node [comp] (c2) at (2,-4) {$C_2$};
            \node [comp] (c3) at (-1.90,0.80) {$C_3$};
            }         
            \begin{scope}[on background layer]
                \draw 
                (3,2) node (v) [norm, label={[label distance=0.05cm]90:\scriptsize $v$}] {}
                (-2,-2) node (u1) [norm, label={[label distance=0.05cm]180:\scriptsize $u_1$}] {}
                (0,-3) node (u2) [norm, label={[label distance=0.08cm]270:\scriptsize $u_2$}] {}
                (2,-2) node (u3) [norm, label={[label distance=0.1cm]-5:\scriptsize $u_3$}] {}
                
                (-2.5,-3.5) node (u12v1) [invisible] {}
                (-1,-4) node (u12v2) [invisible] {}
                
                (2.5,-3.5) node (u32v1) [invisible] {}
                (1,-4) node (u32v2) [invisible] {}
                
                (-1,1) node (u22v1) [invisible] {}
                (-2,0) node (u22v2) [invisible] {}
                (-1.5,0.5) node (u22v3) [invisible] {}
                
                (u1)--(u12v1)
                (u2)--(u12v2)
                
                (u2)--(u32v2)
                (u3)--(u32v1)
                
                (u1)--(u22v2)
                (u2)--(u22v3)
                (u3)--(u22v1)
                
                (v)--(u1)
                (v)--(u2)
                (v)--(u3)
                ;
                \draw[draw=MyColorOne, line width=0.5mm]
                (u1)--(u2)
                ;
            \end{scope}
        \end{tikzpicture}
        \caption{}
        \label{fig:Fv:Gu1u2}
    \end{subfigure}
    \hspace{2.5em}
    \begin{subfigure}[b]{0.15\textwidth}
        \begin{tikzpicture}[scale=0.35, auto,swap]
            {
            \node [comp] (c1) at (-2,-4) {$C_1$};
            \node [comp] (c2) at (2,-4) {$C_2$};
            \node [comp] (c3) at (-1.90,0.80) {$C_3$};
            }         
            \begin{scope}[on background layer]
                \draw 
                (3,2) node (v) [norm, label={[label distance=0.05cm]90:\scriptsize $v$}] {}
                (-2,-2) node (u1) [norm, label={[label distance=0.05cm]180:\scriptsize $u_1$}] {}
                (0,-3) node (u2) [norm, label={[label distance=0.08cm]270:\scriptsize $u_2$}] {}
                (2,-2) node (u3) [norm, label={[label distance=0.1cm]-5:\scriptsize $u_3$}] {}
                
                (-2.5,-3.5) node (u12v1) [invisible] {}
                (-1,-4) node (u12v2) [invisible] {}
                
                (2.5,-3.5) node (u32v1) [invisible] {}
                (1,-4) node (u32v2) [invisible] {}
                
                (-1,1) node (u22v1) [invisible] {}
                (-2,0) node (u22v2) [invisible] {}
                (-1.5,0.5) node (u22v3) [invisible] {}
                
                (u1)--(u12v1)
                (u2)--(u12v2)
                
                (u2)--(u32v2)
                (u3)--(u32v1)
                
                (u1)--(u22v2)
                (u2)--(u22v3)
                (u3)--(u22v1)
                
                (v)--(u1)
                (v)--(u2)
                (v)--(u3)
                ;
                \draw[draw=MyColorOne, line width=0.5mm]
                (u1)--(u2)
                (u2)--(u3)
                ;
            \end{scope}
        \end{tikzpicture}
        \caption{}
        \label{fig:Fv:Gu1u2u2u3}
    \end{subfigure}
    \hspace{2.5em}
    \begin{subfigure}[b]{0.15\textwidth}
        \begin{tikzpicture}[scale=0.35, auto,swap]
            {
            \node [comp] (c1) at (-2,-4) {$C_1$};
            \node [comp] (c2) at (2,-4) {$C_2$};
            \node [comp] (c3) at (-1.90,0.80) {$C_3$};
            }         
            \begin{scope}[on background layer]
                \draw 
                (3,2) node (v) [norm, label={[label distance=0.05cm]90:\scriptsize $v$}] {}
                (-2,-2) node (u1) [norm, label={[label distance=0.05cm]180:\scriptsize $u_1$}] {}
                (0,-3) node (u2) [norm, label={[label distance=0.08cm]270:\scriptsize $u_2$}] {}
                (2,-2) node (u3) [norm, label={[label distance=0.1cm]-5:\scriptsize $u_3$}] {}
                
                (-2.5,-3.5) node (u12v1) [invisible] {}
                (-1,-4) node (u12v2) [invisible] {}
                
                (2.5,-3.5) node (u32v1) [invisible] {}
                (1,-4) node (u32v2) [invisible] {}
                
                (-1,1) node (u22v1) [invisible] {}
                (-2,0) node (u22v2) [invisible] {}
                (-1.5,0.5) node (u22v3) [invisible] {}
                
                (u1)--(u12v1)
                (u2)--(u12v2)
                
                (u2)--(u32v2)
                (u3)--(u32v1)
                
                (u1)--(u22v2)
                (u2)--(u22v3)
                (u3)--(u22v1)
                
                (v)--(u1)
                (v)--(u2)
                (v)--(u3)
                ;
                \draw[draw=MyColorOne, line width=0.5mm]
                (u1)--(u2)
                (u2)--(u3)
                (u1)--(u3)
                ;
            \end{scope}
        \end{tikzpicture}
        \caption{}
        \label{fig:Fv:GFv}
    \end{subfigure}
    
    \caption{
    (a) the graph $G=(V,E)$.
    For simplicity's sake, the vertex set of the connected components $C_1,C_2,C_3$ are grouped.
    $\{u_1,u_2\},\{u_2,u_3\},\{u_1,u_2,u_3\}$ are all minimal separators of $G$ contained within $N_G(v)$.
    Notice that $|\textit{fill}(\{u_1,u_2\})|=
    |\textit{fill}(\{u_2,u_3\})|=1$ but $|\textit{fill}(\{u_1,u_2,u_3\})|=3$.
    (b), (c) respectively depict the graphs $G\oplus\{u_1u_2\}$ and $G\oplus\{u_1u_2,u_2u_3\}$.
    Note that $|\textit{fill}_{G\oplus\{u_1u_2,u_2u_3\}}(\{u_1,u_2,u_3\})|=1$.
    (d) the graph $G\oplus F^v$ where $F^v=\{u_1u_2,u_2u_3,u_1u_3\}$.
    Note that $v$ is a simplicial vertex in $G\oplus F^v$.
    The added edges in all the above graphs are drawn in green.
    }
    \label{alg:Fv}
\end{figure}

{
The following theorem summarises our results in this section and can be used to construct a minimum triangulation for a graph.

\begin{theorem}\label{lem:safeVertexRepeatedly}
Let $G$ be a graph and $\alpha$ be an elimination ordering over $G$ where for a fixed step $k\in\{1,\dots,|V|-1\}$, both of the following conditions are satisfied: 
\begin{enumerate}
    \item[(1)] For every step $i\in\{1,\dots,k\}$, $\alpha(i)$ satisfies one of the following conditions (where $F^{\alpha(i)}$ is the set of edges defined by Algorithm \ref{alg:Fv}):
    \begin{itemize}
        \item[(A)] $\alpha(i)$ is a simplicial vertex in $G_{i-1}$ or $G_{i-1}\oplus F^{\alpha(i)}$.
        \item[(B)] $\textit{deg}_{G_{i-1}}(\alpha(i))
        =\kappa(G_{i-1})$ and $\alpha(i)$ is an almost simplicial vertex in the graph $G_{i-1}$ or $G_{i-1}\oplus F^{\alpha(i)}$.
    \end{itemize}
    
    \item[(2)] $\alpha$ is any minimum elimination ordering over the graph $G_k$.
\end{enumerate}
Then $\alpha$ is a minimum elimination ordering over $G$.
\end{theorem}
\begin{proof}
To prove that $\alpha$ is a minimum elimination ordering, it suffices to prove that the number of edges added by $\alpha$ is $\textit{mfi}(G)$.
For simplicity, let us refer to the set of edges added at step $i\in\{1,\dots,|V|-1\}$ of $\alpha$ with $R_i$, i.e.,\ $R_i=\textit{fill}_{G_{i-1}}\,(N_{G_{i-1}}(\alpha(i)))$.
Notice that by definition, $F^{\alpha(i)}\subseteq R_i$ for every step $i\in\{1,\dots,|V|-1\}$ of $\alpha$.
Then to prove the statement of the theorem, we need to show that
$|\bigcup_{i\in\{1,\dots,|V|\}}R_i|=
\sum_{i\in\{1,\dots,|V|\}}|R_i|
=\textit{mfi}(G)$.

Let $k$ be a fixed step. Since $\alpha$ is a minimum elimination ordering over the graph $G_k$, we have $\textit{mfi}(G_k)=\sum_{i\in\{k+1,\dots,|V|\}}|R_i|$.
Then it only remains to prove the following:
\begin{equation}\label{eq:lem:safeVertexRepeatedly:0}
    \textit{mfi}(G)-
    \sum\nolimits_{i\in\{1,\dots,k\}}|R_i|
    =\textit{mfi}(G_k)
\end{equation}

We begin by showing that in the first step of $\alpha$, $R_1$ is safe to add in $G$.
By selection, $\alpha(1)$ satisfies one of the conditions (A) or (B) stated in the theorem.
Let us consider the following complementary cases in regard to $F^{\alpha(1)}$:
\begin{mycases}
    \case
    {
    Suppose that $F^{\alpha(1)}=\emptyset$.
    Then $\alpha(1)$ is simplicial in $G$, and hence $R_1=\emptyset$ and obviously safe to add, or $\alpha(1)$ is an almost simplicial vertex in $G$ and $\textit{deg}_G(\alpha(1))=\kappa(G)$.
    Then, by \cref{lem:degkAlmostSimplicial}, $R_1$ is safe to add in $G$.
    }
    
    \case
    {
    Suppose that $F^{\alpha(1)}\neq\emptyset$.
    Then either $\alpha(1)$ is simplicial in $G\oplus F^{\alpha(1)}$ implying that $\textit{fill}_G(N(\alpha(1)))=F^{\alpha(1)}$ which in turn means that $R_1$ is safe to add in $G$ or $\alpha(1)$ is almost simplicial in $G\oplus F^{\alpha(1)}$ and $\textit{deg}_G(\alpha(1))=\kappa(G)$.
    In the latter case, we know from Observation \ref{ob:1} that $F^{\alpha(1)}$ is safe to add in $G$ and \cref{lem:degkAlmostSimplicial} shows that $\textit{fill}_{G\oplus F^{\alpha(1)}}(N(\alpha(1)))$ is also safe to add in $G\oplus F^{\alpha(1)}$ consequently, the union of these edge sets, $R_1$, is also safe to add by \cref{lem:safeRepeatedly}.
    }
    
\end{mycases}

As a result of \cref{lem:safeEdgesFirst} we can state that:
\begin{equation}\label{eq:lem:safeVertexRepeatedly:1}
    \textit{mfi}(G)-|R_1|=\textit{mfi}(G\oplus R_1)
\end{equation}

Since $R_1=\textit{fill}_{G}\,(N(\alpha(1)))$, $\alpha(1)$ is simplicial in the graph $G\oplus R_1$.
Consequently, $N_{G\oplus R_1}\,(\alpha(1))$ is a clique minimal separator in $G\oplus R_1$ with components $G\oplus R_1[N[\alpha(1)]]$ and $(G\oplus R_1)-\alpha(1)$.
$G\oplus R_1[N[\alpha(1)]]$ is a complete graph (thus has minimum fill-in of zero) so by \cref{lem:cliqueMinimalSeparator} we can ignore this component:

\begin{equation}\label{eq:lem:safeVertexRepeatedly:2}
    \textit{mfi}(G\oplus R_1)=
    \textit{mfi}((G\oplus R_1)-\alpha(1))
    =\textit{mfi}(G_1)
\end{equation}
 
Then by combining \cref{eq:lem:safeVertexRepeatedly:1} and \cref{eq:lem:safeVertexRepeatedly:2} we can state the following:
\begin{equation}\label{eq:lem:safeVertexRepeatedly:3}
    \textit{mfi}(G)-|R_1|=\textit{mfi}(G_1)
\end{equation}

In the case where $k=1$, the theorem follows as this proves \cref{eq:lem:safeVertexRepeatedly:0}, so assume that $k\geq 2$.
In the following we argue that the following equation holds for any value of $l\in\{2,\dots,k\}$ thus proving \cref{eq:lem:safeVertexRepeatedly:0}.
\begin{equation}
    \textit{mfi}(G_l)-|R_l|=\textit{mfi}(G_l)
\end{equation}

By applying the same argument as above for step 2 of $\alpha$, we can state that:
\begin{equation}\label{eq:lem:safeVertexRepeatedly:4}
    \textit{mfi}(G_1)-|R_2|=\textit{mfi}(G_2)
\end{equation}

This combined with \cref{eq:lem:safeVertexRepeatedly:3} implies that:
\begin{equation}\label{eq:lem:safeVertexRepeatedly:5}
    \textit{mfi}(G)-(|R_1|+|R_2|)=\textit{mfi}(G_2)
\end{equation}

Surely, by repeating this argument for every $l$, until step $k$, we get \cref{eq:lem:safeVertexRepeatedly:0} as required.
The theorem thus follows.

\end{proof}

}
}

\section{$\tau$ Parameter in Relation With Treewidth and Connectivity}\label{sec:tfm:Ktrees}

{
In this section we prove that $\tau$ equals 0 for all graphs  with treewidth of at most two and for graphs, in which treewidth and vertex connectivity have the same value.
Halin graphs are an example of such graphs which are not chordal and have treewidth and vertex connectivity of three.

{
\begin{lemma}\label{lem:tw=kConnected}
Let $G$ be a graph where $\textit{tw}(G)=\kappa(G)$.
Then $\tau(G)=0$.
\end{lemma}

\begin{proof}
Let $\textit{tw}(G)=k$.
Then, there exists an elimination ordering $\beta$ over $G$ where ${\max_{v\in V} |\textit{madj}^+_\beta(v)|} \leq k$.
This together with statement (i) from \cref{cor:tw=kConnected} implies that $|\textit{madj}^+_\beta(\beta(i))| =k$ for every step $i\in\{1,\dots,|V|-k-1\}$.
Combining this with statement (ii) in the same lemma, we have that
${\sum_{v\in V}|\textit{madj}^+_\beta(v)|} =k(|V|-k-1) +\frac{k(k+1)}{2}$.
Also, by \cref{cor:tw=kConnected}, $|E|+\textit{mfi}(G)\geq k(|V|-k-1) +\frac{k(k+1)}{2}=|E_{G^+_\beta}|$.
This shows that $G^+_\beta$ is a minimum triangulation over $G$.
Since $\textit{tw}(G^+_\beta) ={\max_{v\in V} |\textit{madj}^+_\beta(v)|} =k$ we have that $\textit{tw}(G^+_\beta)=\textit{tw}(G)$ proving that $\tau(G)=0$ following \cref{def:tau}.
\end{proof}
}

{
\cref{lem:tw=kConnected} actually demonstrates a stronger property of graphs where treewidth and connectivity have the same value.
This property is stated below:

\begin{corollary}\label{cor:kConnectedtwiffmfi}
Let $G$ be a graph with $\kappa(G)=\textit{tw}(G)$.
Then a minimal triangulation $H$ of $G$ is a minimum triangulation iff $\textit{tw}(H)=\textit{tw}(G)$.
\end{corollary}
\begin{proof}
Let $k=\kappa(G)=\textit{tw}(G)$.
To prove the statement of the corollary in the `\textit{if}' direction, let $H$ be a minimum triangulation of $G$ and $\alpha$ a perfect elimination ordering over $H$.
Notice that by selection, $\alpha$ is a minimum elimination ordering over $G$.
As shown in the proof of \cref{lem:tw=kConnected}, we have
$\sum_{v\in V}|\textit{madj}^+_\alpha(v)|=k(|V|-k-1)+\frac{k(k+1)}{2}|$.
This, combined with the statements from \cref{cor:tw=kConnected}
implies that 
$|\textit{madj}^+_\alpha(\alpha(i))| =k$ for every step $i\in\{1,\dots,|V|-k-1\}$ and
$G^\alpha_{|V|-k-1}$ is a $K_{k+1}$, therefore 
$\max_{v\in V}|\textit{madj}^+_\alpha(v)|
=k=\textit{tw}(H)$.

Now, let us prove the corollary in the `\textit{only if}' direction.
To do so, let $H$ be a fixed minimal triangulation where $\textit{tw}(H)=\textit{tw(G)}$ and let $\beta$ be a perfect elimination ordering over $H$.
Once again, notice that $\beta$ is a minimal elimination ordering over $G$ where $\max_{v\in V}|\textit{madj}^+_\beta(v)|
=\textit{tw}(H)=k$.
By the selection of $\beta$ and statements (i) from \cref{cor:tw=kConnected}, $|\textit{madj}^+_\beta(\beta(i))|=k$ for every step ${i\in\{1,\dots,|V|-k-1\}}$.
This combined with statement (ii) from \cref{cor:tw=kConnected} implies that
$\sum_{v\in V}|\textit{madj}^+_\beta(v)|= k(|V|-k-1)+\frac{k(k+1)}{2} =|E|+\textit{mfi}(G)$ (the last part follows from the proof of  \cref{lem:tw=kConnected}).
In conclusion, we have that $H$ is a minimum triangulation as desired.
\end{proof}
}

\begin{lemma}\label{lem:tw<=2}
Let $G$ be a graph where $\textit{tw}(G)\leq 2$.
Then $\tau(G)=0$.
\end{lemma}

\begin{proof}
Clearly, $\tau(G)=0$ if $\textit{tw}(G)=1$ as trees are chordal.
Additionally, if $\textit{tw}(G)=\kappa(G)=2$ then following \cref{cor:tw=kConnected}, $\tau(G)=0$.
So, it remains to discuss the case where $\textit{tw}(G)=2$ and $\kappa(G)\leq 1$ (note that vertex connectivity of a graph cannot be greater than its treewidth).

Let $\mathcal{Q}$ be the set of all biconnected components of $G$.
We also point out that by definition, $\mathcal{Q}$ is the set of all maximal connected subgraphs of $G$ where every two subgraphs share at most one vertex.
For every component $Q\in\mathcal{Q}$, either (a) $Q$ is a complete graph on two vertices, or (b) $\kappa(Q)=2$.

We point out that every cycle of the graph $G$ is entirely contained within the biconnected components $Q$ of type (b) from $\mathcal{Q}$.
Therefore, given any minimum triangulation $H=(V,E\cup F)$ of $G$, for every edge $uv\in F$ both vertices $u,v$ are contained in some biconnected component $Q$ of type (b).

Now we define a minimum triangulation $H=(V,E\cup F)$
as it follows.
Let $Q=(V_Q,E_Q)$ be a component of type (b) from $\mathcal{Q}$.
Since $Q$ is a subgraph of $G$,  $\textit{tw}(Q)=2$ and then obviously also $\kappa(Q)=2$.
Following the proof of \cref{lem:tw=kConnected}, there exists a minimum triangulation $Q^*=(V_Q,E_Q\cup F_Q)$ of $Q$ where $\textit{tw}(Q^*)=2$.
We add $F_Q$ to $F$ and repeat the same process for every other component $Q$ of type (b) in $\mathcal{Q}$.

By construction, $H=(V,E\cup F)$ is a minimum triangulation of $G$ as 
$F$ is the set of all added edges that construct minimum triangulations over all components of type (b) of $\mathcal{Q}$.
Additionally, because every maximal clique of the triangulation $H=(V,E\cup F)$ is contained within one of the triangulations $Q^*$ defined above, we have $\textit{tw}(H)=\textit{tw}(G)=2$.
This concludes that $\tau(G)=0$ as required.
\end{proof}
}

\section{Future Work}\label{sec:conc}
Future work will be carried out to characterise graph classes with $\tau$ of 0, e.g., an the extension 
of \cref{lem:tw<=2} for graphs with treewidth of at most 3 or other well-studied graph classes.
These methods could also potentially be used to determine some lower/upper bound for minimum fill-in of some graph classes.

Additionally, we will be looking for other conditions for an edge set in a graph to be safe to add as well as generalisation of 
\cref{lem:degkAlmostSimplicial}, e.g., if a requirement on the degree and connectivity can be relaxed.

\bibliography{mybibfile}

\end{document}